\documentclass[11pt]{article}

\usepackage{times}
\usepackage[paperwidth=199.8mm,paperheight=297mm,centering,
            hmargin=20mm,vmargin=20mm]{geometry}
\usepackage{authblk}
\usepackage[bottom]{footmisc}
\usepackage[titletoc,title]{appendix}
\usepackage[T1]{fontenc}
\usepackage[utf8]{inputenc}
\usepackage{amsmath,amssymb,amsthm,mathtools}
\usepackage{microtype,booktabs,array}
\usepackage[shortlabels]{enumitem}
\usepackage{needspace,aliascnt}
\usepackage[dvipsnames]{xcolor}
\usepackage{tcolorbox}
\tcbuselibrary{skins,breakable}

\definecolor{myblue}{RGB}{0,68,116}
\definecolor{definitionframe}{RGB}{112,112,112}
\definecolor{definitionback}{RGB}{248,248,248}

\usepackage{hyperref}
\hypersetup{
  colorlinks=true,
  citecolor=myblue,
  linkcolor=myblue,
  filecolor=myblue,
  urlcolor=myblue,
  breaklinks=true
}
\usepackage[nameinlink,noabbrev]{cleveref}

\newtheorem{theorem}{Theorem}

\newaliascnt{lemma}{theorem}
\newtheorem{lemma}[lemma]{Lemma}
\aliascntresetthe{lemma}
\newaliascnt{proposition}{theorem}
\newtheorem{proposition}[proposition]{Proposition}
\aliascntresetthe{proposition}
\newaliascnt{corollary}{theorem}
\newtheorem{corollary}[corollary]{Corollary}
\aliascntresetthe{corollary}
\newtheoremstyle{boldremark}
  {6pt}{6pt}{\normalfont}{}{\bfseries}{.}{0.5em}{}
\theoremstyle{boldremark}
\newaliascnt{remark}{theorem}
\newtheorem{remark}[remark]{Remark}
\aliascntresetthe{remark}
\newtheorem*{remark*}{Remark}

\tcbset{
  supplement result/.style={
    enhanced,
    unbreakable,
    arc=2pt,
    outer arc=2pt,
    boxrule=0.55pt,
    left=6pt,
    right=6pt,
    oversize=0pt,
    top=6pt,
    bottom=6pt,
    before skip=9pt plus 2pt minus 1pt,
    after skip=9pt plus 2pt minus 1pt
  }
}
\newenvironment{mainresult}{%
  \begin{tcolorbox}[supplement result,
    colframe=myblue,colback=myblue!5!white]%
}{%
  \end{tcolorbox}%
}

\newcommand{\Tr}{\operatorname{Tr}}
\newcommand{\id}{\operatorname{id}}
\newcommand{\Var}{\operatorname{Var}}
\newcommand{\E}{\mathbb E}
\newcommand{\ket}[1]{|#1\rangle}
\newcommand{\bra}[1]{\langle#1|}

\allowdisplaybreaks[1]
\title{\Large\textbf{Approximate majorization and high-order capacity of \\ quantum depolarizing channels}}
\author[1]{Zhiwei Song \thanks{zhiweisong@cuhk.edu.cn}}
\author[2]{Xin Wang \thanks{felixxinwang@hkust-gz.edu.cn}}
\author[3,4]{Marco Tomamichel \thanks{marco.tomamichel@nus.edu.sg}}
\author[1]{Kun Fang \thanks{kunfang@cuhk.edu.cn}}

\affil[1]{\small{School of Data Science, The Chinese University of Hong Kong, Shenzhen,
Guangdong, 518172, China}}
\affil[2]{\small{Thrust of Artificial Intelligence, Information Hub, The Hong Kong University of Science \protect \\ and Technology (Guangzhou), Guangzhou 511453, China}}
\affil[3]{\small{Department of Electrical and Computer Engineering,
National University of Singapore,\protect \\ Singapore 117583, Singapore}}
\affil[4]{\small{Centre for Quantum Technologies, National University of Singapore, Singapore 117543, Singapore}}

\begin{document}
\date{\today}
\maketitle

\begin{abstract}
A central challenge in quantum information theory is to determine whether entanglement provides an advantage for information processing when coherent operations across multiple channel uses are allowed. For depolarizing channels, King's proof of the additivity of the minimum output entropy rules out such an advantage for classical communication at leading asymptotic order, while leaving open the possibility of an advantage for higher orders. In this work, we conjecture a majorization relation for the outputs of depolarizing channels and prove an approximate version of this conjecture. Our proof combines an exact ordering of shifted log-determinants with an anti-concentration bound for the spectrum of the product-state output, refining King's analysis based on Schatten norms. An additional logarithmic second-moment estimate controls arbitrary channel outputs. As a direct application, we obtain a tight third-order expansion for the classical capacity of quantum depolarizing channels. Consequently, entangled codewords and collective decoding can improve the logarithm of the message size over the product-coding strategy by at most
$O(1)$ bits, sharpening our understanding of the role of entanglement in the finite-blocklength regime.
\end{abstract}

{\tableofcontents}

\section{Introduction}
Understanding how noise limits information processing is of central interest to
quantum information theory. Even the simple and highly symmetric quantum depolarizing channel leaves fundamental questions unresolved, including the quantum communication capacity~\cite{SSWR17,LDS18,JD24,ZZW24,zhu2025geometric}, distillable entanglement of its Choi state~\cite{LDS18,zhu2025geometric}, and the exact tensorization of hypercontractive and logarithmic Sobolev inequalities in general~\cite{KT13,Kin14,TPK14,BDR20,DGO26}. A central challenge underlying these questions is to determine whether entanglement provides any advantage in the corresponding tasks when coherent operations across multiple channel uses are allowed.

For transmitting classical messages through quantum depolarizing channels, King established the leading-order capacity by proving the multiplicativity of the maximal output Schatten norms and, consequently, the additivity of the minimum output entropy~\cite{Kin03}. Together with the Holevo--Schumacher--Westmoreland coding theorem~\cite{SW97,Hol98}, these results show that the capacity can be achieved by encoding messages into product states and by performing product measurements in the corresponding basis. Thus, neither entangled signal states nor entangled measurements provide an advantage for classical communication at leading asymptotic order.

Nevertheless, this leading-order optimality leaves open the role of entanglement in finite-blocklength communication. At a fixed error tolerance, entangled codewords might still increase the number of reliably transmissible messages through subleading corrections or accelerate the convergence of communication rates to capacity as the blocklength increases. Resolving this possibility requires a more precise understanding of the spectral properties of the channel outputs. A natural, albeit strong, conjecture is that the channel outputs obey a majorization ordering. This would strengthen King's multiplicativity result for Schatten norms to a statement covering all convex, unitarily invariant functionals.

Specifically, for a $d$-dimensional system, define the quantum depolarizing channel by
\begin{align}
 {\mathcal D}_{p,d}(X):=(1-p)X+\frac pd\Tr(X)I_d,\qquad
 0\le p\le p_{\max}(d):=\frac{d^2}{d^2-1},\quad d\ge2.
 \label{eq:channel}
\end{align} 
We conjecture that the output corresponding to a pure product input majorizes the output corresponding to any other input. More precisely, let $\rho^n$ denote an arbitrary, possibly entangled, input state, and let $\ket{0^n}:=\ket{0}^{\otimes n}$ denote a fixed pure product input. Then
\begin{align}
 \lambda^\downarrow\!\left({\mathcal D}_{p,d}^{\otimes n}(\rho^n)\right)
 \prec
 \lambda^\downarrow\!\left({\mathcal D}_{p,d}^{\otimes n}(\ket{0^n}\bra{0^n})\right).
 \label{eq:output-majorization-conjecture}
\end{align}
Here, $\lambda^\downarrow$ lists the eigenvalues in nonincreasing order, and $x\prec y$ means that $x$ is majorized by $y$. 

Under this conjecture, every channel output could be obtained from the pure
product reference output ${\mathcal D}_{p,d}^{\otimes n}(\ket{0^n}\bra{0^n})$ by a mixed-unitary channel~\cite{Wat18}.
Since such channels preserve the maximally mixed state, data processing would
make the reference output optimal for hypothesis testing against that state.
The Wang--Renner meta-converse~\cite{WR12} would then turn this comparison
into a one-shot upper bound on the communication rate. Because the reference and maximally mixed
states commute and have tensor-product form, the meta-converse bound reduces to a classical hypothesis testing problem. Combining its known third-order
expansion~\cite{Tan14} with orthogonal-basis coding achievability~\cite{Pol10}
would determine the high-order communication rates.

Extensive numerical experiments support the majorization conjecture in~\eqref{eq:output-majorization-conjecture}, but proving it exactly remains challenging and requires
control of every leading eigenvalue sum for arbitrary entangled inputs. For $0\le p\le1$, the conjecture was proved for $n=2$~\cite{FH02}, and for arbitrary $n$, is a special case of the stronger spin-alignment conjecture~\cite{AK24}. Although a counterexample to the latter  was constructed recently~\cite{SC26}, this does not rule out conjecture~\eqref{eq:output-majorization-conjecture} for $n\ge 3$, which incorporates an additional compatibility
constraint. In this work, we prove an approximate version of the full majorization that is sufficiently accurate for obtaining the third-order asymptotics of the communication rates. The approximation is quantified through the spectral
defect $\delta_{\operatorname{maj}}(X,Y)$, which measures the largest
excess of a leading eigenvalue sum of $X$ over the corresponding sum of
$Y$. Equivalently, it is the minimum total variation distance from the spectrum of $X$ to a probability vector majorized by the spectrum of $Y$~\cite{HOS18}. This defect vanishes exactly when $X$ is majorized by $Y$.

\begin{mainresult}
\begin{theorem}[Approximate majorization]\label{thm:approximation-intro}
Fix $d\ge2$ and $0<p\le p_{\max}(d)$ with $p\ne1$.
There exists $\kappa_{d,p}>0$ such that, for every $n\ge1$
and every state $\rho^n$ on $(\mathbb C^d)^{\otimes n}$,
\begin{align}
 \delta_{\operatorname{maj}}\!\left({\mathcal D}_{p,d}^{\otimes n}(\rho^n),{\mathcal D}_{p,d}^{\otimes n}(\ket{0^n}\bra{0^n})\right)
 \le\delta_n:=\frac{\kappa_{d,p}}{\sqrt n}.
 \label{eq:approximation-intro}
\end{align}
The constant $\kappa_{d,p}$ is independent of $n$ and $\rho^n$.
\end{theorem}
\end{mainresult}

Our proof begins with an exact ordering of shifted log-determinants between
arbitrary channel outputs and the pure product reference:
\begin{align}
  \ln \det(tI + {\mathcal D}_{p,d}^{\otimes n}(\rho^n)) \geq  \ln \det(tI + {\mathcal D}_{p,d}^{\otimes n}(\ket{0^n}\bra{0^n})), \quad \forall t > 0.
\end{align} 
This ordering is a necessary condition for the full majorization conjectured in~\eqref{eq:output-majorization-conjecture} since the log-determinant is a concave, unitarily invariant spectral functional. At the same time, it is sufficient to recover King's Schatten-norm ordering for $1<q\leq 2$. More importantly, by rescaling the log-determinant and choosing the shift parameter $t$ appropriately, we obtain an approximate majorization relation. For a fixed i.i.d.\ reference with positive spectral-information variance, the general bound is $O(\log\log n/\sqrt n)$ and is asymptotically tight. As shown in Appendix~\ref{app:norm-defect}, such approximate majorization cannot be obtained from Schatten-norm ordering alone. This necessitates the stronger log-determinant ordering, which provides finer control over the output spectra.
For depolarizing outputs, an additional logarithmic second-moment estimate improves the bound to $O_{d,p}(n^{-1/2})$.

As a direct application, the approximate majorization 
yields a tight third-order expansion for the classical capacity of quantum
depolarizing channels. Let $M^*(n,\varepsilon)$ be the largest number of
messages transmissible over $n$ channel uses with average error at most
$\varepsilon$. Arbitrary entangled codewords and
collective decoding are allowed. 

\begin{mainresult}
\begin{theorem}[Third-order classical capacity]\label{thm:main}
For $d\ge2$, $0<p\le p_{\max}(d)$ with $p\ne1$, and
$\varepsilon\in(0,1)$,
\begin{align}
 \log M^*(n,\varepsilon)
 =nc_{d,p}+\sqrt{nv_{d,p}}\,\Phi^{-1}(\varepsilon)
       +\frac12\log n+O_{d,p,\varepsilon}(1)
 \label{eq:main}
\end{align}
as $n\to\infty$, where $\Phi$ is the standard normal cumulative distribution function. The classical capacity $c_{d,p}$ and
dispersion $v_{d,p}$ are given explicitly later in Eq.~\eqref{eq:parameters}.
\end{theorem}
\end{mainresult}

This extends King's capacity formula by determining both the Gaussian and logarithmic corrections at fixed error tolerance. Encoding and measuring in an orthonormal basis achieves all three terms, up to a bounded remainder~\cite{Pol10}. Consequently, entangled codewords and collective decoding can improve the logarithm of the message size over this product-coding strategy by at most
$O_{d,p,\varepsilon}(1)$, thereby sharpening our understanding of the advantage provided by entanglement in the finite-blocklength regime.

\section{Preliminaries}\label{sec:preliminaries}

\subsection{Notation and definitions}\label{sec:notation}
All Hilbert spaces are finite dimensional. We write $I_m$ for the identity
on an $m$-dimensional space and $\pi_D:=I_D/D$ for the maximally mixed
state on a $D$-dimensional space.
The identity map is denoted by $\id$. A state is a positive
semidefinite operator of trace one. An effect is an operator $0\le Q\le I$.
The trace is unnormalized. We use $\ln$ for natural logarithms and
$\log $ for logarithms in base two.
We write $S(X):=-\Tr(X\log X)$ for the von Neumann entropy, with
$0\log 0:=0$. For $X\ge0$, the Schatten norms are
$\|X\|_q:=(\Tr X^q)^{1/q}$ for $1\le q<\infty$ and
$\|X\|_\infty:=\lambda_{\max}(X)$.

For any single-use pure input, the channel output has eigenvalues
$\alpha_p:=1-p+p/d$ and $d-1$ copies of $\beta_p:=p/d$.
We take $\ket{0}$ to be the first computational-basis vector.
The capacity and dispersion are
\begin{align}
 c_{d,p}&:=\log d+\alpha_p\log \alpha_p
              +(d-1)\beta_p\log \beta_p,\nonumber\\
 v_{d,p}&:=\alpha_p(1-\alpha_p)
       \left(\log  \alpha_p - \log  \beta_p\right)^2.
 \label{eq:parameters}
\end{align}
At $p=0$, set $v_{d,0}:=0$ by continuity. For
$0<p\le p_{\max}(d)$ with $p\ne1$, both $c_{d,p}$ and $v_{d,p}$
are strictly positive. Our asymptotic statements keep the channel
parameters and error tolerance fixed.\footnote{Under the average-error
convention, the degenerate cases $p=0,1$ have optimal message sizes
$\lfloor d^n/(1-\varepsilon)\rfloor$ and
$\lfloor1/(1-\varepsilon)\rfloor$, respectively. The remainder in
Eq.~\eqref{eq:main} is not claimed to be uniform near degenerate
channel or error parameters.}

\subsection{Approximate majorization}

Write $\lambda^\downarrow(X)$ for the eigenvalues of a state $X$ in
nonincreasing order. For real vectors $x,y\in\mathbb R^m$, let
$x^\downarrow$ and $y^\downarrow$ be their nonincreasing rearrangements.
We write $x\prec y$ if
$\sum_{j=1}^k x_j^\downarrow\le\sum_{j=1}^k y_j^\downarrow$ for
$1\le k<m$ and $\sum_{j=1}^m x_j=\sum_{j=1}^m y_j$.
Thus $x\prec y$ means that $x$ is majorized by $y$.

For a $D$-dimensional state $X$ with eigenvalues
$\lambda_1(X)\ge\cdots\ge\lambda_D(X)$, let
$\mathcal K_k(X):=\sum_{j=1}^k\lambda_j(X)$ for $k=1,\ldots,D$, with
$\mathcal K_0(X):=0$. For $k\le u\le k+1$ and $k=0,\ldots,D-1$, define
\begin{align}
 \mathcal K_u(X):=\mathcal K_k(X)+(u-k)\lambda_{k+1}(X).
 \label{eq:kyfan-interpolation}
\end{align}
Thus $\mathcal K_u(X)$ linearly interpolates the consecutive integer partial sums. It is nondecreasing and concave in $u$.
Its variational characterization is
\begin{align}
 \mathcal K_u(X)=\max_{\substack{0\le Q\le I_D\\ \Tr Q\le u}}\Tr(QX),
 \qquad 0\le u\le D.
 \label{eq:kyfan-variational}
\end{align}

For two states on the same space, define the one-sided spectral defect
\begin{align}
 \delta_{\operatorname{maj} }(X,Y):=\max_{0\le k\le D}\{\mathcal K_k(X)-\mathcal K_k(Y)\}.
 \label{eq:majorization-defect}
\end{align}

The maximum is over integer $k$ and is nonnegative because
$\mathcal K_0(X)=\mathcal K_0(Y)=0$. Linear interpolation gives
$\mathcal K_u(X)\le \mathcal K_u(Y)+\delta_{\operatorname{maj}}(X,Y)$ for every $u\in[0,D]$.
In particular, $\delta_{\operatorname{maj}}(X,Y)=0$ means that the spectrum of $Y$
majorizes that of $X$.
Equivalently, $\delta_{\operatorname{maj}}(X,Y)$ is the largest vertical
excess of the upper Lorenz curve of $X$ above that of $Y$.

\section{Approximate majorization for depolarizing outputs}

An orthonormal basis $\{\ket{u_i}\}_{i=1}^d$ is \emph{uniform}
relative to the computational basis if
$|\langle j|u_i\rangle|^2=1/d$ for all $i,j$.
Set $P_i:=\ket{u_i}\bra{u_i}$ and define the
\emph{phase-damping channel}
\begin{align}
 \Delta_{p,d}(X):=(1-p)X+p\sum_{i=1}^dP_iXP_i,
 \qquad 0\le p\le\frac{d}{d-1}.
 \label{eq:phase}
\end{align}
This channel preserves diagonal entries and scales off-diagonal entries
by $1-p$ in the indicated basis. King's decomposition
\cite{Kin03} expresses the depolarizing
channel as a convex combination of unitary conjugates of uniform
phase-damping channels. This decomposition will also be a key ingredient
in the log-determinant comparison below.
Specifically, for every $0\le p\le p_{\max}(d)$, there exist nonnegative
weights $c_a$, unitaries $V_a$, and uniform bases
$\{\ket{u_i^{(a)}}\}_{i=1}^d$, indexed by
$a=1,\ldots,2d^2(d+1)$, with $\sum_a c_a=1$, such that
\begin{align}
 {\mathcal D}_{p,d}(X)=\sum_{a=1}^{2d^2(d+1)} c_aV_a^\dagger\Delta_{p,d}^{(a)}(X)V_a
 \label{eq:king}
\end{align}
for every operator $X$, where $\Delta_{p,d}^{(a)}$ denotes the phase-damping
channel in Eq.~\eqref{eq:phase} associated with the $a$th uniform basis.
The nonnegative coefficients are obtained by regrouping the
identity-phase contribution and may vanish at the endpoints.

\subsection{Log-determinant ordering for depolarizing outputs}\label{sec:spectral}

For a positive semidefinite operator $X$ and $t>0$, define
\begin{align}
 \mathcal L_t(X):=\ln\det(tI+X),
 \label{eq:logdet-definition}
\end{align}
where $I$ acts on the same space as $X$.
The map $\mathcal L_t$ is concave and unitarily invariant
\cite{BV04}; the positive shift ensures that it is finite
even when $X$ is singular. For states, $\mathcal L_t$ also admits an interpretation in terms of
the $f_s$-divergence from the maximally mixed state.
For $s>0$, let $f_s(u):=\ln(1+s)-\ln(s+u)$ for $u\ge0$, and for a $D$-dimensional state $\omega$,
\begin{align}
 D_{f_s}(\omega\|\pi_D)
 =\ln(1+s)-\frac1D\ln\det(sI_D+D\omega)=\ln\frac{1+s}{D}-\frac1D\mathcal L_{s/D}(\omega).
 \label{eq:logdet-divergence}
\end{align}

The following lemma bounds the log-determinant of the output of the phase-damping channel solely in terms of its block-diagonal part. It parallels a bound in King's analysis~\cite{Kin03}. However, the key ingredients here are Sylvester's determinant identity and the concavity of the log-determinant, rather than the Lieb--Thirring inequality used in King's proof.

\begin{lemma}[Block determinant inequality]\label{lem:block}
Let $\Delta_{p,d}$ be a phase-damping channel with respect to the uniform basis $\{\ket{u_i}\}_{i=1}^d$. Let $R\ge0$ act on $\mathbb C^d\otimes\mathbb C^m$ and $R_{ij} = (\bra{u_i} \otimes I)R(|u_j\rangle\otimes I)$ be the $(i,j)$-th block.
For $0\le p\le p_{\max}(d)$ and $t>0$,
\begin{align}
 \mathcal L_t\bigl((\Delta_{p,d}\otimes\id)(R)\bigr)
 &\ge\frac1d\sum_{i=1}^d\big[
 \mathcal L_t(d\alpha_pR_{ii})
 +(d-1)\mathcal L_t(d\beta_pR_{ii})\big].
 \label{eq:block}
\end{align}
\end{lemma}

\begin{proof}
Set $B_i:=tI_m+pR_{ii}>0$ and
$B:=\sum_{i=1}^d\ket{u_i}\bra{u_i}\otimes B_i$.
Define the normalized block columns
$T_i:=R^{1/2}(\ket{u_i}\otimes B_i^{-1/2})$. Then
\begin{align}
 \mathcal L_t\bigl((\Delta_{p,d}\otimes\id)(R)\bigr)
 &=\ln\det B+
 \ln\det\!\left(I_{dm}+(1-p)\sum_{i=1}^dT_iT_i^\dagger\right)\\
 &\ge\ln\det B+\frac1d\sum_{i=1}^d
 \ln\det\!\left(I_{dm}+d(1-p)T_iT_i^\dagger\right)\\
 &=\sum_{i=1}^d\ln\det B_i+\frac1d\sum_{i=1}^d
 \ln\frac{\det(tI_m+d\alpha_pR_{ii})}{\det B_i}\\
 &=\frac1d\sum_{i=1}^d\bigl[
 \mathcal L_t(d\alpha_pR_{ii})
 +(d-1)\mathcal L_t(d\beta_pR_{ii})\bigr].
 \label{eq:block-normalized-chain}
\end{align}
The first equality follows from
$tI_{dm}+(\Delta_{p,d}\otimes\id)(R)=B+(1-p)R$,
factorization by $B^{1/2}$, and Sylvester's determinant identity.
The inequality is concavity of $\ln\det$.
Both the required positivity and the determinant ratio follow from
\begin{align}
 I_m+d(1-p)T_i^\dagger T_i
 =B_i^{-1/2}(tI_m+d\alpha_pR_{ii})B_i^{-1/2}>0.
 \label{eq:block-normalized-identity}
\end{align}
Here $T_i^\dagger T_i=B_i^{-1/2}R_{ii}B_i^{-1/2}$ and
$p+d(1-p)=d\alpha_p$; positivity follows from $t>0$ and $\alpha_p>0$.
The matrices $T_i^\dagger T_i$ and $T_iT_i^\dagger$ share their nonzero
eigenvalues, so $I_{dm}+d(1-p)T_iT_i^\dagger>0$ as well,
including when $p>1$.
Sylvester's identity then gives the determinant ratio in the chain,
and $d\beta_p=p$ yields the final equality. This concludes the proof.
\end{proof}

We then combine the result of Lemma~\ref{lem:block} with the decomposition in Eq.~\eqref{eq:king} to establish the ordering of the depolarizing outputs by mathematical induction.

\begin{theorem}[Log-determinant ordering of the depolarizing outputs]\label{thm:spectral}
For every $n\ge1$, every quantum state $\rho^n$ on $(\mathbb C^d)^{\otimes n}$,
$t>0$, and $0<p\le p_{\max}(d)$,
\begin{align}
 \mathcal L_t\bigl({\mathcal D}_{p,d}^{\otimes n}(\rho^n)\bigr)
 \ge\mathcal L_t\bigl({\mathcal D}_{p,d}^{\otimes n}(\ket{0^n}\bra{0^n})\bigr).
 \label{eq:spectral}
\end{align}

\end{theorem}

\begin{proof}
Let $\tau:={\mathcal D}_{p,d}(\ket{0}\bra{0})$.
We prove the claim by induction on $n$.
For $n=1$, concavity reduces the minimum to pure inputs, whose outputs
are unitarily equivalent to $\tau$ by covariance.

Assume the claim for $n-1$, with $n\ge2$, and put $m:=d^{n-1}$.
As in King's proof~\cite{Kin03}, covariance implies that applying a unitary to the first input conjugates the
channel output and leaves $\mathcal L_t$ unchanged.
We may therefore assume that the first-input marginal $\rho_A$ is
diagonal in the computational basis.
For every $a=1,\ldots,2d^2(d+1)$ and $i=1,\ldots,d$, uniformity gives
$\langle u_i^{(a)}|\rho_A|u_i^{(a)}\rangle
=\Tr\rho_A/d=1/d$.
Thus
 $\rho_{i,a}^{n-1}
 :=d(\bra{u_i^{(a)}}\otimes I_m)\rho^n
       (\ket{u_i^{(a)}}\otimes I_m)$
is a positive operator of trace one on the remaining $n-1$ inputs.
Here $a$ indexes the bases in Eq.~\eqref{eq:king}, and $i$ indexes
vectors within each basis.
Set $R:=(\id\otimes{\mathcal D}_{p,d}^{\otimes(n-1)})(\rho^n)$.
Its diagonal blocks in the $a$th basis satisfy
\begin{align}
 R_{ii}^{(a)}
 &:=(\bra{u_i^{(a)}}\otimes I_m)R(\ket{u_i^{(a)}}\otimes I_m)
 =\frac1d{\mathcal D}_{p,d}^{\otimes(n-1)}(\rho_{i,a}^{n-1}).
 \label{eq:conditional-output-blocks}
\end{align}
We obtain
\begin{align}
 \mathcal L_t\bigl({\mathcal D}_{p,d}^{\otimes n}(\rho^n)\bigr)
 &=\mathcal L_t\bigl((\mathcal D_{p,d}\otimes\id)(R)\bigr)\\
 &\overset{\text{(a)}}{\ge}
 \sum_{a=1}^{2d^2(d+1)}c_a
 \mathcal L_t\bigl((\Delta_{p,d}^{(a)}\otimes\id)(R)\bigr)\\
 &\overset{\text{(b)}}{\ge}
 \frac1d\sum_{a=1}^{2d^2(d+1)}\sum_{i=1}^d c_a\bigl[
 \mathcal L_t(d\alpha_pR_{ii}^{(a)})
 +(d-1)\mathcal L_t(d\beta_pR_{ii}^{(a)})\bigr]
 \label{eq:induct}\\
 &\overset{\text{(c)}}{\ge}
 \mathcal L_t\bigl(\alpha_p\tau^{\otimes(n-1)}\bigr)
 +(d-1)\mathcal L_t\bigl(\beta_p\tau^{\otimes(n-1)}\bigr)\\
 &\overset{\text{(d)}}{=}\mathcal L_t\bigl(\tau^{\otimes n}\bigr).
 \label{eq:product-induction-reassembly}
\end{align}

Step~(a) applies Eq.~\eqref{eq:king} to the first tensor factor,
followed by concavity and unitary invariance of $\mathcal L_t$.
Step~(b) applies Lemma~\ref{lem:block} in each of the indexed uniform bases.
For step~(c), let $c\in\{\alpha_p,\beta_p\}$. Then
\begin{align}
 \mathcal L_t(dcR_{ii}^{(a)})
 &=m\ln c+
 \mathcal L_{t/c}\bigl({\mathcal D}_{p,d}^{\otimes(n-1)}(\rho_{i,a}^{n-1})\bigr) \ge m\ln c+\mathcal L_{t/c}\bigl(\tau^{\otimes(n-1)}\bigr)
 =\mathcal L_t\bigl(c\tau^{\otimes(n-1)}\bigr).
\end{align}
The first equality uses Eq.~\eqref{eq:conditional-output-blocks} and
$\mathcal L_t(cX)=m\ln c+\mathcal L_{t/c}(X)$ for $X\ge0$ on an $m$-dimensional space.
The inequality is the induction hypothesis at the positive shift $t/c$;
both $\alpha_p$ and $\beta_p$ are positive throughout the stated range.
The final equality uses the same scaling identity.
Averaging this bound gives step~(c), since
$\frac1d\sum_{a=1}^{2d^2(d+1)}\sum_{i=1}^d c_a
=\sum_{a=1}^{2d^2(d+1)}c_a=1$.
Step (d) uses the block-diagonal form of
$\tau\otimes\tau^{\otimes(n-1)}$: one block is
$\alpha_p\tau^{\otimes(n-1)}$, and the other $d-1$ blocks are
$\beta_p\tau^{\otimes(n-1)}$.
This completes the induction.
\end{proof}

\medskip
\begin{remark}
King proved multiplicativity of the maximal output $q$-norm for a
depolarizing channel tensored with an arbitrary auxiliary channel
\cite{Kin03}. In particular,
\begin{align}
 \|{\mathcal D}_{p,d}^{\otimes n}(\rho^n)\|_q
 \le\|{\mathcal D}_{p,d}^{\otimes n}(\ket{0^n}\bra{0^n})\|_q,
 \qquad 1\le q\le\infty.
 \label{eq:king-norm-comparison}
\end{align}
The log-determinant comparison in Theorem~\ref{thm:spectral} also
implies Eq.~\eqref{eq:king-norm-comparison} for $1<q\le2$, and hence
\begin{align}
 S\bigl({\mathcal D}_{p,d}^{\otimes n}(\rho^n)\bigr)
 &\ge S({\mathcal D}_{p,d}^{\otimes n}(\ket{0^n}\bra{0^n})).
 \label{eq:shared-entropy-comparison}
\end{align}
The proof is given in Appendix~\ref{app:entropy-comparison}.
Thus the two approaches share the entropy comparison underlying the
first-order capacity formula. Only the limit $q\to1^+$ is needed
for this consequence. The log-determinant ordering does not,
however, imply the full family of higher-norm inequalities for arbitrary
states, as the $q=3$ example in Appendix~\ref{app:entropy-comparison} shows.
\end{remark}

\subsection{Approximate majorization from log-determinant ordering}\label{sec:approximation}

We next bound the gap between Lorenz curves under log-determinant ordering.
For $x\ge0$ and $t>0$, let $\phi_t(x):=t\ln(1+x/t)$ be the scalar
perspective of $\ln(1+x)$ \cite{BV04}. For fixed $t>0$, $\phi_t$ is increasing and concave, with
$0\le\phi_t(x)\le x$.
For a positive semidefinite operator $X$ on a $D$-dimensional space,
define the rescaled log-determinant functional
\begin{align}
 \mathcal P_t(X)&:=\Tr\phi_t(X)
 =t\ln\det(I_D+X/t)
 =t\bigl[\mathcal L_t(X)-D\ln t\bigr].
 \label{eq:perspective-definition}
\end{align}
This shows that $\mathcal P_t$ and
$\mathcal L_t$ give equivalent orderings in a fixed dimension.

To compare Lorenz ordinates, we decompose $\mathcal P_t$ into spectral
tail mass and two nonnegative correction terms.
For a state $\omega$ with nonincreasing eigenvalues $z_1,\ldots,z_D$
and $k=0,\ldots,D$,
\begin{align}
 \mathcal P_t(\omega)
 =1-\mathcal K_k(\omega)
 +\underbrace{\sum_{j=1}^k\phi_t(z_j)}_{\text{head contribution}}
 -\underbrace{\sum_{j=k+1}^D\bigl[z_j-\phi_t(z_j)\bigr]}_{\text{tail loss}}.
 \label{eq:perspective-tail-decomposition}
\end{align}
Here $1-\mathcal K_k(\omega)$ is the complement of the Lorenz ordinate at $k$.
Under $\mathcal P_t(\rho)\ge\mathcal P_t(\sigma)$, nonnegativity of
the correction terms bounds $\mathcal K_k(\rho)-\mathcal K_k(\sigma)$
by the head contribution of $\rho$ and the tail loss of $\sigma$.

We control these terms using the concentration of the reference
state's spectral information. For a real random variable $Z$ and
$h>0$, the L\'evy concentration function is defined by \cite{HT73}
\begin{align}
 Q(Z,h):=\sup_{a\in\mathbb R}\Pr\{a<Z\le a+h\}.
 \label{eq:spectral-concentration}
\end{align}
For a state $\sigma$ with eigenvalues $y_1,\ldots,y_D$, define
$S_\sigma:=-\log y_J$, where $\Pr\{J=j\}=y_j$.
Zero eigenvalues are never sampled.
Thus $Q(S_\sigma,1)$ is the largest probability in any interval
of width one bit.
For a state $\omega$, let $N_\omega(u)$ denote the number of its
eigenvalues at least $u>0$.

\begin{proposition}[Approximate majorization from log-determinant ordering]\label{prop:logdet-majorization}
Let $\rho$ and $\sigma$ be states on the same $D$-dimensional space,
and suppose that $\mathcal L_t(\rho)\ge\mathcal L_t(\sigma)$ for every $t>0$.
Then
\begin{align}
 \delta_{\operatorname{maj}}(\rho,\sigma)
 \le q\left[1+2\ln\left(1+\ln\frac eq\right)\right]\qquad \text{with}\quad q:=Q(S_\sigma,1).
 \label{eq:optimized-majorization-bound}
\end{align}
\end{proposition}
\begin{proof}
Write $x_1\ge\cdots\ge x_D\ge0$ and $y_1\ge\cdots\ge y_D\ge0$
for the eigenvalues of $\rho$ and $\sigma$, respectively.
For $1\le k<D$ with $y_k>0$, Eq.~\eqref{eq:perspective-tail-decomposition}
gives, for every $t>0$,
\begin{align}
 \mathcal K_k(\rho)-\mathcal K_k(\sigma)
 &\le\sum_{j=1}^k[\phi_t(x_j)-\phi_t(y_j)]
      +\sum_{j=k+1}^D[y_j-\phi_t(y_j)]\\
 &\le\sum_{j=1}^k\phi_t(x_j)-k\phi_t(y_k)
      +\sum_{j=k+1}^D[y_j-\phi_t(y_j)]\\
 &=\sum_{j=1}^k\phi_t(x_j)-k\phi_t(y_k)
      +\int_0^{y_k}\frac{u[N_\sigma(u)-k]}{t+u}\,du\\
 &=\sum_{j=1}^k\phi_t(x_j)-ky_k
      +\int_0^{y_k}\frac{uN_\sigma(u)}{t+u}\,du.
 \label{eq:defect-split}
\end{align}

The first inequality discards
$\mathcal P_t(\sigma)-\mathcal P_t(\rho)$ and
$-\sum_{j=k+1}^D[x_j-\phi_t(x_j)]$, both of which are nonpositive.
The second uses the monotonicity of $\phi_t$.
The first equality uses the integral equality
$y_j-\phi_t(y_j)=\int_0^{y_j}u/(t+u)\,du$.
For $0<u\le y_k$, exactly $N_\sigma(u)-k$ tail eigenvalues are at least
$u$, and none exceeds $y_k$.
The second equality uses the integral equality again.

We next bound the integral in Eq.~\eqref{eq:defect-split}. Set $a:=-\log u$ and observe that
\begin{align}
 uN_\sigma(u)
 &=\sum_{j:y_j\ge u}u
 =\sum_{j:y_j>0}y_j\frac{u}{y_j}\mathbf1_{\{y_j\ge u\}}\\
 &=\mathbb E\!\left[\frac{u}{y_J}\mathbf1_{\{y_J\ge u\}}\right]
 =\mathbb E\!\left[
 2^{S_\sigma-a}\mathbf1_{\{S_\sigma\le a\}}\right]\\
 &=\sum_{r=0}^{\infty}\mathbb E\!\left[
 2^{S_\sigma-a}\mathbf1_{\{a-r-1<S_\sigma\le a-r\}}\right]\\
 &\le\sum_{r=0}^{\infty}\mathbb E\!\left[
 2^{-r}\mathbf1_{\{a-r-1<S_\sigma\le a-r\}}\right]\\
 &=\sum_{r=0}^{\infty}2^{-r}
 \Pr\{a-r-1<S_\sigma\le a-r\}\\
 &\le q\sum_{r=0}^{\infty}2^{-r}=2q.
 \label{eq:threshold-cardinality}
\end{align}
The first line uses the definition of $N_\sigma(u)$ and writes each
contribution $u$ as $y_j(u/y_j)$, with the indicator selecting $y_j\ge u$.
The sum is restricted to $y_j>0$ so that every ratio is well defined.
The second line uses $\Pr\{J=j\}=y_j$, followed by
$2^{S_\sigma-a}=u/y_J$ and the equivalence
$S_\sigma\le a\Longleftrightarrow y_J\ge u$.
The third line partitions $(-\infty,a]$ into the disjoint intervals
$(a-r-1,a-r]$, $r=0,1,\ldots$.
On each corresponding event, $S_\sigma-a\le-r$, so
$2^{S_\sigma-a}\le2^{-r}$; this gives the first inequality.
The next equality uses the fact that the expectation of an indicator
is the probability of its event.
The last inequality uses $q=Q(S_\sigma,1)$ on each interval,
and the final equality sums the geometric series.

Now take $t=y_k/L$, where $L=\ln(e/q)\ge1$. Then
\begin{align}
 \mathcal K_k(\rho)-\mathcal K_k(\sigma)
 &\le\sum_{j=1}^k\phi_t(x_j)-ky_k+2q\ln(1+L)\\
 &\le q\sum_{j=1}^k x_j+kt\ln(1/q)-ky_k+2q\ln(1+L)\\
 &=q\sum_{j=1}^k x_j-kt+2q\ln(1+L)\\
 &\le q+2q\ln(1+L).
\end{align}
The first inequality combines Eq.~\eqref{eq:defect-split} with
Eq.~\eqref{eq:threshold-cardinality}.
The second inequality uses $\phi_t(x_j)\le qx_j+t\ln(1/q)$, which follows from
$\ln(1+r)\le qr-\ln q$ for $r\ge0$ and $0<q\le1$.
The equality and following inequality can be verified directly. 
The bound is uniform in $k$.
If $y_k=0$, then $\mathcal K_k(\sigma)=1$, so the defect is nonpositive.
The cases $k=0,D$ are immediate. Maximizing over $k$ proves
Eq.~\eqref{eq:optimized-majorization-bound}.
\end{proof}

In particular, let $\rho_n$ and $\sigma_n$ be states on the same
$D_n$-dimensional space for each $n$, satisfying
$\mathcal L_t(\rho_n)\ge\mathcal L_t(\sigma_n)$ for every $t>0$.
If $q_n:=Q(S_{\sigma_n},1)\to0$, then
\begin{align}
 \delta_{\operatorname{maj}}(\rho_n,\sigma_n)
 =O\!\left(q_n\left[1+\ln\ln\frac{e}{q_n}\right]\right)\to0,
 \label{eq:general-asymptotic-majorization}
\end{align}
with a universal implied constant and no tensor-product assumption.
This anti-concentration condition is sufficient, not necessary.
Log-determinant ordering alone does not ensure a vanishing defect:
taking the constant sequences $\rho_n=X$ and $\sigma_n=Y$ from
Eq.~\eqref{eq:logdet-higher-norm-example} gives
$\delta_{\operatorname{maj}}(\rho_n,\sigma_n)=1/50$ for every $n$.

For an independent and identically distributed (i.i.d.) reference, the
spectral information is a sum of
independent random variables. A positive single-copy variance gives
the following uniform bound.

\begin{mainresult}
\begin{theorem}[Approximate majorization for i.i.d. reference states]\label{thm:approximation}
Let $\sigma$ be a fixed state on $\mathbb C^d$ with
$V_\sigma:=\operatorname{Var}(S_\sigma)>0$.
There exists $\kappa_\sigma>0$ such that, for all sufficiently large $n$
and every state $\rho^n$ on $(\mathbb C^d)^{\otimes n}$ satisfying
$\mathcal L_t(\rho^n)\ge\mathcal L_t(\sigma^{\otimes n})$ for every $t>0$,
\begin{align}
 \delta_{\operatorname{maj}}(\rho^n,\sigma^{\otimes n})
 \le\frac{\kappa_\sigma}{\sqrt n}(\log \log n+\frac12).
 \label{eq:iid-approximation-bound}
\end{align}
The constant $\kappa_\sigma$ and the starting blocklength depend only
on $\sigma$.
\end{theorem}
\end{mainresult}
\begin{proof}
Let $S_n$ be the sum of $n$ independent copies of $S_\sigma$.
It has the same distribution as $S_{\sigma^{\otimes n}}$, with
mean $nS(\sigma)=n\mathbb E(S_\sigma)$ and variance $nV_\sigma$.
Let $F_n$ be the cumulative distribution function of
$(S_n-nS(\sigma))/\sqrt{nV_\sigma}$, and set
$B_\sigma:=3\mathbb E|S_\sigma-S(\sigma)|^3/V_\sigma^{3/2}$.
This constant is finite because $S_\sigma$ has finite support.
We obtain
\begin{align}
 q_n:=Q(S_n,1)
 &=\sup_{b\in\mathbb R}\Pr\{b<S_n\le b+1\}\\
 &=\sup_{b\in\mathbb R}\left[
 F_n\!\left(\frac{b+1-nS(\sigma)}{\sqrt{nV_\sigma}}\right)
 -F_n\!\left(\frac{b-nS(\sigma)}{\sqrt{nV_\sigma}}\right)\right]\\
 &=\sup_{a\in\mathbb R}\left[
 F_n\!\left(a+\frac1{\sqrt{nV_\sigma}}\right)-F_n(a)\right]\\
 &\le\sup_{a\in\mathbb R}\left[
 \Phi\!\left(a+\frac1{\sqrt{nV_\sigma}}\right)-\Phi(a)\right]
 +\frac{2B_\sigma}{\sqrt n}\\
 &\le\frac1{\sqrt n}\left[\frac1{\sqrt{2\pi V_\sigma}}+2B_\sigma\right]
 =\frac{\kappa_\sigma}{2\sqrt n}.
 \label{eq:unit-window}
\end{align}
The first equality is the definition of the L\'evy concentration function.
The second expresses the interval probability using the distribution
function $F_n$ of the standardized sum.
The third changes variables to $a=(b-nS(\sigma))/\sqrt{nV_\sigma}$,
which ranges over all of $\mathbb R$ as $b$ does because $V_\sigma>0$.
Thus standardization sends the unit interval $(b,b+1]$ for $S_n$
to $(a,a+1/\sqrt{nV_\sigma}]$: subtracting the mean shifts the interval,
and dividing by the standard deviation scales its length from $1$
to $1/\sqrt{nV_\sigma}$.
The first inequality applies the Berry--Esseen bound
$\sup_a|F_n(a)-\Phi(a)|\le B_\sigma/\sqrt n$ (e.g.,
\cite{Dur19}) at both endpoints, accounting for the factor two.
The second inequality uses $\Phi'(a)\le1/\sqrt{2\pi}$, so the Gaussian
probability of an interval is at most its length divided by $\sqrt{2\pi}$.
The final equality defines
$\kappa_\sigma:=2/\sqrt{2\pi V_\sigma}+4B_\sigma$,
which depends only on $\sigma$.
Applying Proposition~\ref{prop:logdet-majorization} to $\rho^n$ and
$\sigma^{\otimes n}$ gives, for all sufficiently large $n$,
\begin{align}
 \delta_{\operatorname{maj}}(\rho^n,\sigma^{\otimes n})
 &\le\frac{\kappa_\sigma}{\sqrt n}
 \left[\frac12+\ln\left(1+\ln\frac{2e\sqrt n}{\kappa_\sigma}\right)\right]
 \le\frac{\kappa_\sigma}{\sqrt n}(\log \log n+\frac12).
 \label{eq:kyfan-final-bound}
\end{align}
The first inequality uses Eq.~\eqref{eq:unit-window} and the monotonicity
of the right-hand side of Eq.~\eqref{eq:optimized-majorization-bound}
for $0<q\le1$.
The second follows from
$1+\ln(2e\sqrt n/\kappa_\sigma)\le\log n$ for all sufficiently large $n$.
All estimates depend only on $\sigma$, so the bound is uniform over
$\rho^n$.
\end{proof}

The bound in Theorem~\ref{thm:approximation} is asymptotically tight.
For the fixed reference $\sigma=\operatorname{diag}(2/3,1/3)$,
Appendix~\ref{app:sharpness} constructs full-rank diagonal states
$\rho^n$ satisfying the log-determinant hypothesis, with
$\delta_{\operatorname{maj}}(\rho^n,\sigma^{\otimes n})
=\Theta((\log\log n)/\sqrt n)$.
For depolarizing outputs, however, we obtain the stronger
$O(n^{-1/2})$ bound by exploiting their additional spectral properties.

\begin{corollary}[Approximate majorization for depolarizing outputs]\label{cor:depolarizing-approximation}
Fix $d\ge2$ and $0<p\le p_{\max}(d)$ with $p\ne1$.
There exists $\kappa_{d,p}>0$ such that, for every $n\ge1$
and every state $\rho^n$ on $(\mathbb C^d)^{\otimes n}$,
\begin{align}
 \delta_{\operatorname{maj}}\!\left({\mathcal D}_{p,d}^{\otimes n}(\rho^n),{\mathcal D}_{p,d}^{\otimes n}(\ket{0^n}\bra{0^n})\right)
 \le\frac{\kappa_{d,p}}{\sqrt n}.
 \label{eq:approximation-bound}
\end{align}
The constant $\kappa_{d,p}$ is independent of $n$ and $\rho^n$.
\end{corollary}
\begin{proof}
Let $\tau:={\mathcal D}_{p,d}(\ket{0}\bra{0})$,
$X:={\mathcal D}_{p,d}^{\otimes n}(\rho^n)$, and
$Y:=\tau^{\otimes n}$.
The spectral information $S_\tau$ takes the values
$-\log \alpha_p$ and $-\log \beta_p$ with probabilities
$\alpha_p$ and $1-\alpha_p$, respectively.
Hence $V_\tau=v_{d,p}>0$.
Theorem~\ref{thm:spectral} gives
$\mathcal L_t(X)\ge\mathcal L_t(Y)$ for all $t>0$.
Write $x_1\ge\cdots\ge x_{d^n}$ and $y_1\ge\cdots\ge y_{d^n}$
for the eigenvalues of $X$ and $Y$, respectively.
For $1\le k<d^n$, put $t=y_k>0$.
Equation~\eqref{eq:defect-split} gives
\begin{align}
 \sum_{j=1}^k x_j-\sum_{j=1}^k y_j
 &\le\sum_{j=1}^k[\phi_t(x_j)-t]
       +\int_0^t\frac{uN_Y(u)}{t+u}\,du\\
 &\le\frac t4\sum_{j=1}^k\ln^2(1+x_j/t)
       +\int_0^t\frac{uN_Y(u)}{t+u}\,du\\
 &\le\frac t4\Tr\!\left[\ln^2(I+X/t)\right]
       +\int_0^t\frac{uN_Y(u)}{t+u}\,du.
 \label{eq:threshold-defect-bound}
\end{align}
The second inequality uses $r-1\le r^2/4$ with $r=\ln(1+x_j/t)$.
The last extends the sum of nonnegative squared logarithms to all eigenvalues.

We estimate the logarithmic second moment and the reference integral separately.
For the reference $Y$, Eqs.~\eqref{eq:threshold-cardinality} and
\eqref{eq:unit-window}, with $\sigma=\tau$, give, for every $u>0$ and $n\ge1$,
\begin{align}
 uN_Y(u)
 \le\frac{\kappa_\tau}{\sqrt n}.
 \label{eq:cor9-reference-count}
\end{align}
For $X$, Lemma~\ref{lem:output-truncated-moment} in
Appendix~\ref{app:truncated-moment} provides a logarithmic second-moment
bound with a constant $B_{d,p}>0$ depending only on $d,p$.
Set $t_0:=2^{-n[S(\tau)+c_{d,p}/2]}$, where
$c_{d,p}=\log d-S(\tau)>0$.
For every $t\ge t_0>d^{-n}$, the lemma gives
\begin{align}
 t\Tr\!\left[\ln^2(I+X/t)\right]
 \le\frac{B_{d,p}\sqrt n}{\ln(d^nt)}\le\frac{C_{d,p}}{\sqrt n},
 \label{eq:cor9-output-moment}
\end{align}
where $C_{d,p}:=2B_{d,p}/(c_{d,p}\ln2)$.
The last inequality uses $\ln(d^nt)\ge nc_{d,p}\ln2/2>0$.

If $y_k\ge t_0$, substituting Eqs.~\eqref{eq:cor9-output-moment}
and \eqref{eq:cor9-reference-count} into
Eq.~\eqref{eq:threshold-defect-bound} yields
\begin{align}
 \sum_{j=1}^k x_j-\sum_{j=1}^k y_j
 \le\frac{C_{d,p}/4+\kappa_\tau\ln2}{\sqrt n}.
 \label{eq:cor9-large-threshold}
\end{align}
If $y_k<t_0$, then
\begin{align}
 \sum_{j=1}^k x_j-\sum_{j=1}^k y_j
 &\le\sum_{j>k}y_j
 \le\sum_{j:y_j<t_0}y_j\\
 &=\Pr\{S_Y>-\log t_0\}\\
 &=\Pr\{S_Y-nS(\tau)>nc_{d,p}/2\}\le\frac{4v_{d,p}}{nc_{d,p}^{\,2}}.
 \label{eq:cor9-small-threshold}
\end{align}
The last inequality is Chebyshev's inequality, since
$S_Y$ has mean $nS(\tau)$ and variance $nv_{d,p}$.
The endpoint cases $k=0,d^n$ have zero defect. Therefore,
\begin{align}
 \delta_{\operatorname{maj}}(X,Y)
 \le\frac{\kappa_{d,p}}{\sqrt n},
\end{align}
where $\kappa_{d,p}:=\max\{4v_{d,p}/c_{d,p}^{\,2},
C_{d,p}/4+\kappa_\tau\ln2\}$ depends only on $d,p$. This completes the proof.
\end{proof}

\begin{remark}
The vanishing defect in Eq.~\eqref{eq:approximation-bound} cannot be deduced
from King's Schatten-norm bounds in Eq.~\eqref{eq:king-norm-comparison} alone.
Appendix~\ref{app:norm-defect} constructs full-rank states satisfying
all these bounds for $1\le q\le\infty$ relative to the same product reference,
while their spectral defects tend to $1/2$. This shows that the log-determinant ordering provides a refined spectral control
used in the finite-blocklength converse.
\end{remark}

\section{High-order capacity by approximate majorization}\label{sec:auxiliary-lemmas}

As a direct application of the approximate majorization of depolarizing outputs, we derive a third-order expansion for the classical capacity of the quantum depolarizing channel, with an $O(1)$ remainder. This refined spectral control goes beyond first-order asymptotics, showing that entanglement provides no advantage, not only for the asymptotic capacity itself, but also for the Gaussian and logarithmic corrections.

We first specify the coding setup. An $(n,M)$ code consists of input
states $\rho_1^n,\ldots,\rho_M^n$ on $(\mathbb C^d)^{\otimes n}$ and a
decoding positive operator-valued measure (POVM) $\{\Lambda_x\}_{x=1}^M$
on the output, with $\Lambda_x\ge0$ and
$\sum_{x=1}^M\Lambda_x=I_{d^n}$.
Its average error is
\begin{align}
 P_{\rm err}:=1-\frac1M\sum_{x=1}^M
 \Tr\bigl[\Lambda_x{\mathcal D}_{p,d}^{\otimes n}(\rho_x^n)\bigr].
 \label{eq:error}
\end{align}
The maximal message size $M^*(n,\varepsilon)$ is the largest $M$
with $P_{\rm err}\le\varepsilon$.
Inputs may be entangled across all channel uses and decoding may be
collective.

\subsection{One-shot converse by approximate majorization}

The following lemma combines the hypothesis-testing meta-converse of Wang and Renner~\cite{WR12} with the approximate majorization of the channel outputs to establish a one-shot converse. We provide a self-contained proof for completeness.

For a state $X$, a positive reference $Y$ on the same space, and
$0\le\varepsilon<1$, define
\begin{align}
 \beta_\varepsilon(X\|Y):=
 \min_{\substack{0\le Q\le I\\ \Tr(QX)\ge1-\varepsilon}}\Tr(QY),
 \qquad
 D_H^\varepsilon(X\|Y):=-\log \beta_\varepsilon(X\|Y).
 \label{eq:testing-definition}
\end{align}

\begin{lemma}[One-shot converse by approximate majorization]\label{lem:testing}
Let $\mathcal N$ be a finite-dimensional quantum channel with a
$D$-dimensional output space, and let $\omega$ be a state on that space.
Suppose that $0\le\varepsilon<1$, $0\le\delta<1-\varepsilon$, and
$\max_\rho\delta_{\operatorname{maj}}\bigl(\mathcal N(\rho),\omega\bigr)\le\delta$, where the
maximum is over all input states. Then every $M$-message code for
$\mathcal N$ with average error at most $\varepsilon$ satisfies
\begin{align}
 \log  M\le D_H^{\varepsilon+\delta}(\omega\|\pi_D).
 \label{eq:transfer}
\end{align}
\end{lemma}

\begin{proof}
Consider any such code with input states $\rho_1,\ldots,\rho_M$ and
POVM $\{\Lambda_x\}_{x=1}^M$. We have
\begin{align}
 1-\varepsilon
 & \le\frac1M\sum_{x=1}^M
       \Tr\!\bigl[\Lambda_x\mathcal N(\rho_x)\bigr]\\
    & \le\frac1M\sum_{x=1}^M
       \mathcal K_{\Tr \Lambda_x}(\mathcal N(\rho_x))\\
 & \le\frac1M\sum_{x=1}^M \mathcal K_{\Tr\Lambda_x}(\omega)+\delta
 \le \mathcal K_{D/M}(\omega)+\delta.
 \label{eq:decoder-lorenz-average}
\end{align}
The first inequality is the error constraint, and the second follows
from Eq.~\eqref{eq:kyfan-variational}. The third uses the spectral assumption;
the last uses
concavity of $\mathcal K_u(\cdot)$ in $u$ and $\sum_{x=1}^M\Tr\Lambda_x=D$.
For any state $\omega$ on the output space, Eq.~\eqref{eq:kyfan-variational}
and the definition of $\beta_\varepsilon$ give
\begin{align}
 \beta_\varepsilon(\omega\|\pi_D)
 =\frac1D\min\{u\in[0,D]:\mathcal K_u(\omega)\ge1-\varepsilon\}.
 \label{eq:testing-kyfan}
\end{align}
Applying this relation to Eq.~\eqref{eq:decoder-lorenz-average} yields
$\beta_{\varepsilon+\delta}(\omega\|\pi_D)\le1/M$, proving
Eq.~\eqref{eq:transfer}.
\end{proof}

\subsection{Proof of the third-order formula in Theorem~\ref{thm:main}}\label{sec:mainproof}

\paragraph{Achievability.}
Encoding and measurement in a fixed orthonormal basis induce
the classical $d$-ary symmetric channel
\begin{align}
 W_{p,d}(y|x):=
 \begin{cases}
 \alpha_p,&y=x,\\
 \beta_p,&y\ne x,
 \end{cases}
 \qquad x,y\in\{0,\ldots,d-1\}.
 \label{eq:classical}
\end{align}
Every classical code for $W_{p,d}$ gives a product-state quantum code
with the same average error.
By symmetry, uniform input achieves capacity $c_{d,p}$ and induces uniform output.
Moreover, $P_{X|Y=y}(x)=W_{p,d}(y|x)$, so the conditional
information-density distribution is independent of $y$.
Thus the dispersion and reverse dispersion both equal $v_{d,p}>0$.
The symmetric-channel theorem
\cite{Pol10} therefore gives
\begin{align}
 \log M^*(n,\varepsilon)
 \ge nc_{d,p}+\sqrt{nv_{d,p}}\Phi^{-1}(\varepsilon)
       +\frac12\log n+O_{d,p,\varepsilon}(1).
 \label{eq:achieve}
\end{align}
\paragraph{Converse.}
Let $D:=d^n$ and $\tau:={\mathcal D}_{p,d}(\ket{0}\bra{0})$.
By Eq.~\eqref{eq:approximation-bound}, we have $\delta_n\to0$. Since
$\varepsilon\in(0,1)$ is fixed, $\varepsilon+\delta_n<1$ for all
sufficiently large $n$. We have
\begin{align}
 \log M^*(n,\varepsilon)
 &\le D_H^{\varepsilon+\delta_n}
      \left(\tau^{\otimes n}\middle\|\pi_D\right)
 \label{eq:converse-testing}\\
 &\le nc_{d,p}+\sqrt{nv_{d,p}}\,\Phi^{-1}(\varepsilon+\delta_n)
      +\frac12\log n+O_{d,p,\varepsilon}(1)
 \label{eq:converse-testing-expansion}\\
 &=nc_{d,p}+\sqrt{nv_{d,p}}\,\Phi^{-1}(\varepsilon)
      +\frac12\log n+O_{d,p,\varepsilon}(1).
 \label{eq:converse-final}
\end{align}
The first inequality follows from Corollary~\ref{cor:depolarizing-approximation}
and Lemma~\ref{lem:testing}.
For the second inequality, dephasing a test in the eigenbasis of
$\tau^{\otimes n}$ reduces the problem to classical i.i.d.
hypothesis testing. Under the null distribution, the one-letter
log-likelihood ratio takes the values $\log (d\alpha_p)$ and
$\log (d\beta_p)$ with probabilities $\alpha_p$ and $1-\alpha_p$.
Its mean and variance are $c_{d,p}$ and $v_{d,p}>0$, and its third
absolute centered moment is finite. The third-order expansion in
\cite{Tan14} therefore applies.
Its proof gives a uniform $O(1)$ remainder for error tolerances in
any compact subinterval of $(0,1)$: the Berry--Esseen corrections
are of order $n^{-1/2}$, and $(\Phi^{-1})'$ is bounded there.
This permits the choice $\varepsilon+\delta_n\to\varepsilon$.
With base-two logarithms, the divergence in \cite{Tan14} equals ours
plus $\log (1-\varepsilon-\delta_n)$; this uniformly bounded term
is absorbed into the remainder.

Finally, Eq.~\eqref{eq:converse-final} follows from
$\Phi^{-1}(\varepsilon+\delta_n)=\Phi^{-1}(\varepsilon)
+O_\varepsilon(\delta_n)$ and
$\sqrt n\,\delta_n=\kappa_{d,p}=O_{d,p}(1)$.
All constants are independent of the code. Together with
Eq.~\eqref{eq:achieve}, this proves Theorem~\ref{thm:main}.
\hfill$\square$

\section{Discussion}\label{sec:discussion}

We establish a uniform approximate-majorization order for quantum
depolarizing outputs through exact log-determinant ordering and
spectral anti-concentration, together with a logarithmic second-moment
estimate for arbitrary outputs. Applied to the coding converse, this
order yields the third-order capacity expansion with an $O(1)$ remainder through classical
i.i.d. hypothesis testing and orthogonal-basis achievability.
These results provide spectral control beyond Schatten-norm bounds
and extend the leading-order optimality of orthogonal-basis coding
to the Gaussian and logarithmic corrections. Any advantage of
entangled codewords and collective decoding is therefore limited to
$O_{d,p,\varepsilon}(1)$ bits in the logarithm of the message size. Beyond depolarizing channels, our approach also extends to unital
qubit channels. The log-determinant ordering remains valid in this
setting and, combined with the general approximate-majorization
theorem, yields a third-order capacity expansion with an
$O_{p,\varepsilon}(\log\log n)$ remainder; see Appendix~\ref{app:unital-qubit}.

Several interesting questions remain open for future study. A central one is to establish the full majorization relation in Eq.~\eqref{eq:output-majorization-conjecture} in general. Another important direction is to characterize the corresponding fourth-order behavior: the bounded remainder does not identify a constant-order term or exclude a bounded coding advantage. Beyond these asymptotic refinements, approximate majorization could find applications in other channel models and in quantum-information tasks governed by spectral constraints. Finally, extending the log-determinant and anti-concentration techniques to relative majorization may broaden their applicability to problems involving comparisons between pairs of quantum states. 

\bigskip
\paragraph{Acknowledgements.}
K.F. and Z.S. are supported in part by the National Natural Science Foundation of China (Grants No. 92470113 and 12404569), the Shenzhen Science and Technology Program (Grants No. QNXMB20250701091826036 and JCYJ20240813113519025), the Shenzhen Fundamental Research Program (Grant No. JCYJ20241202124023031), the General R\&D Projects of 1+1+1 CUHK-CUHK(SZ)-GDST Joint Collaboration Fund (Grant No. GRDP2025-022), the Guangdong Provincial Quantum Science Strategic Initiative (Grant No. GDZX2503001), and the University Development Fund (Grant No. UDF01003565). 
XW is supported by the National Natural Science Foundation of China (Grant. No.~92576114, 12447107), the Guangdong Provincial Quantum Science Strategic Initiative (Grant No.~GDZX2403008), and the Guangdong Provincial Key Lab of Integrated Communication, Sensing and Computation for Ubiquitous Internet of Things (Grant No.~2023B1212010007).
MT is supported by the National Research Office, through the NRF Investigatorship award (NRF-NRFI10-2024-0006) and by the Ministry of Education, Singapore, through grant T2EP20124-0005.

OpenAI GPT-6 Astra was used to explore proof strategies and support the writing of the manuscript. We
have verified AI-assisted material and take full responsibility for the content of this work.

\bibliographystyle{alpha}
\bibliography{bib-updated}

@article{SW97,
  author  = {Schumacher, Benjamin and Westmoreland, Michael D.},
  title   = {Sending classical information via noisy quantum channels},
  journal = {Physical Review A},
  year    = {1997},
  volume  = {56},
  number  = {1},
  pages   = {131--138}
}

@article{Hol98,
  author        = {Holevo, Alexander S.},
  title         = {The capacity of the quantum channel with general signal states},
  journal       = {IEEE Transactions on Information Theory},
  year          = {1998},
  volume        = {44},
  number        = {1},
  pages         = {269--273},
  eprint        = {quant-ph/9611023},
  archivePrefix = {arXiv}
}

@article{WR12,
  author        = {Wang, Ligong and Renner, Renato},
  title         = {One-shot classical--quantum capacity and hypothesis testing},
  journal       = {Physical Review Letters},
  year          = {2012},
  volume        = {108},
  number        = {20},
  pages         = {200501},
  eprint        = {1007.5456},
  archivePrefix = {arXiv},
  primaryClass  = {quant-ph}
}

@article{Kin03,
  author        = {King, Christopher},
  title         = {The capacity of the quantum depolarizing channel},
  journal       = {IEEE Transactions on Information Theory},
  year          = {2003},
  volume        = {49},
  number        = {1},
  pages         = {221--229},
  eprint        = {quant-ph/0204172},
  archivePrefix = {arXiv}
}

@article{TPK14,
  author        = {Temme, Kristan and Pastawski, Fernando and Kastoryano, Michael J.},
  title         = {Hypercontractivity of quasi-free quantum semigroups},
  journal       = {Journal of Physics A: Mathematical and Theoretical},
  year          = {2014},
  volume        = {47},
  number        = {40},
  pages         = {405303},
  eprint        = {1403.5224},
  archivePrefix = {arXiv},
  primaryClass  = {quant-ph}
}

@book{BV04,
  author    = {Boyd, Stephen and Vandenberghe, Lieven},
  title     = {Convex Optimization},
  publisher = {Cambridge University Press},
  address   = {Cambridge},
  year      = {2004}
}

@book{Dur19,
  author    = {Durrett, Rick},
  title     = {Probability: Theory and Examples},
  edition   = {5th},
  publisher = {Cambridge University Press},
  year      = {2019}
}

@phdthesis{Pol10,
  author = {Polyanskiy, Yury},
  title  = {Channel Coding: Non-Asymptotic Fundamental Limits},
  school = {Princeton University},
  year   = {2010}
}

@article{Tan14,
  author        = {Tan, Vincent Y. F.},
  title         = {Asymptotic Estimates in Information Theory with Non-Vanishing Error Probabilities},
  journal       = {Foundations and Trends in Communications and Information Theory},
  year          = {2014},
  volume        = {11},
  number        = {1--2},
  pages         = {1--184},
  eprint        = {1504.02608},
  archivePrefix = {arXiv},
  primaryClass  = {cs.IT}
}

@book{HT73,
  author    = {Hengartner, W. and Theodorescu, R.},
  title     = {Concentration Functions},
  series    = {Probability and Mathematical Statistics},
  volume    = {20},
  publisher = {Academic Press},
  address   = {New York--London},
  year      = {1973}
}

@misc{DGO26,
  author        = {Dong, Yangjing and Gao, Li and Ou, Fengning and Yao, Penghui and Zhou, Haigang},
  title         = {Dimension-Free Approximate Tensorization of Quantum Hypercontractivity for Qudit Depolarizing Semigroups},
  howpublished  = {arXiv:2606.17729 [quant-ph]},
  year          = {2026},
  eprint        = {2606.17729},
  archivePrefix = {arXiv},
  primaryClass  = {quant-ph}
}

@article{zhu2025geometric,
  title         = {Geometric optimization for quantum communication},
  author        = {Zhu, Chengkai and Mao, Hongyu and Fang, Kun and Wang, Xin},
  journal       = {arXiv preprint arXiv:2509.15106},
  year          = {2025},
  eprint        = {2509.15106},
  archivePrefix = {arXiv},
  primaryClass  = {quant-ph}
}

@article{SSWR17,
  author        = {Sutter, David and Scholz, Volkher B. and Winter, Andreas and Renner, Renato},
  title         = {Approximate degradable quantum channels},
  journal       = {IEEE Transactions on Information Theory},
  year          = {2017},
  volume        = {63},
  number        = {12},
  pages         = {7832--7844},
  eprint        = {1412.0980},
  archivePrefix = {arXiv},
  primaryClass  = {quant-ph}
}

@article{LDS18,
  author        = {Leditzky, Felix and Datta, Nilanjana and Smith, Graeme},
  title         = {Useful states and entanglement distillation},
  journal       = {IEEE Transactions on Information Theory},
  year          = {2018},
  volume        = {64},
  number        = {7},
  pages         = {4689--4708},
  eprint        = {1701.03081},
  archivePrefix = {arXiv},
  primaryClass  = {quant-ph}
}

@article{JD24,
  author        = {Jabbour, Michael G. and Datta, Nilanjana},
  title         = {Tightening continuity bounds for entropies and bounds on quantum capacities},
  journal       = {IEEE Journal on Selected Areas in Information Theory},
  year          = {2024},
  volume        = {5},
  pages         = {645--658},
  eprint        = {2310.17329},
  archivePrefix = {arXiv},
  primaryClass  = {quant-ph}
}

@article{ZZW24,
  author        = {Zhu, Chengkai and Zhu, Chenghong and Wang, Xin},
  title         = {Estimate distillable entanglement and quantum capacity by squeezing useless entanglement},
  journal       = {IEEE Journal on Selected Areas in Communications},
  year          = {2024},
  volume        = {42},
  number        = {7},
  pages         = {1850--1860},
  eprint        = {2303.07228},
  archivePrefix = {arXiv},
  primaryClass  = {quant-ph}
}

@article{BDR20,
  author        = {Beigi, Salman and Datta, Nilanjana and Rouz{\'e}, Cambyse},
  title         = {Quantum reverse hypercontractivity: Its tensorization and application to strong converses},
  journal       = {Communications in Mathematical Physics},
  year          = {2020},
  volume        = {376},
  number        = {2},
  pages         = {753--794},
  eprint        = {1804.10100},
  archivePrefix = {arXiv},
  primaryClass  = {quant-ph}
}

@article{KT13,
  author        = {Kastoryano, Michael J. and Temme, Kristan},
  title         = {Quantum logarithmic {Sobolev} inequalities and rapid mixing},
  journal       = {Journal of Mathematical Physics},
  year          = {2013},
  volume        = {54},
  number        = {5},
  pages         = {052202},
  eprint        = {1207.3261},
  archivePrefix = {arXiv},
  primaryClass  = {quant-ph}
}

@article{Kin14,
  author        = {King, Christopher},
  title         = {Hypercontractivity for semigroups of unital qubit channels},
  journal       = {Communications in Mathematical Physics},
  year          = {2014},
  volume        = {328},
  number        = {1},
  pages         = {285--301},
  eprint        = {1210.8412},
  archivePrefix = {arXiv},
  primaryClass  = {quant-ph}
}

@book{Wat18,
  author    = {Watrous, John},
  title     = {The Theory of Quantum Information},
  publisher = {Cambridge University Press},
  year      = {2018}
}

@article{Kin02Unital,
  author  = {Christopher King},
  title   = {Additivity for unital qubit channels},
  journal = {Journal of Mathematical Physics},
  volume  = {43},
  number  = {10},
  pages   = {4641--4653},
  year    = {2002},
  doi     = {10.1063/1.1500791},
  eprint  = {quant-ph/0103156},
  archivePrefix = {arXiv},
  note    = {\href{https://doi.org/10.1063/1.1500791}{doi:10.1063/1.1500791}.
             The phase-damping decomposition is Lemma~4 in
             \href{https://arxiv.org/abs/quant-ph/0103156v2}{arXiv:quant-ph/0103156v2}}
}

@article{HOS18,
  author        = {Horodecki, Micha{\l} and Oppenheim, Jonathan and Sparaciari, Carlo},
  title         = {Extremal distributions under approximate majorization},
  journal       = {Journal of Physics A: Mathematical and Theoretical},
  year          = {2018},
  volume        = {51},
  number        = {30},
  pages         = {305301},
  doi           = {10.1088/1751-8121/aac87c},
  eprint        = {1706.05264},
  archivePrefix = {arXiv},
  primaryClass  = {quant-ph}
}

@article{AK24,
  author        = {Alhejji, Mohammad A. and Knill, Emanuel},
  title         = {Towards a Resolution of the Spin Alignment Problem},
  journal       = {Communications in Mathematical Physics},
  year          = {2024},
  volume        = {405},
  number        = {5},
  pages         = {119},
  doi           = {10.1007/s00220-024-04980-1},
  eprint        = {2307.06894},
  archivePrefix = {arXiv},
  primaryClass  = {quant-ph}
}

@article{FH02,
  author  = {Fujiwara, Akio and Hashizum{\'e}, Takashi},
  title   = {Additivity of the capacity of depolarizing channels},
  journal = {Physics Letters A},
  year    = {2002},
  volume  = {299},
  number  = {5--6},
  pages   = {469--475},
  doi     = {10.1016/S0375-9601(02)00735-1}
}

@article{SC26,
  author        = {Song, Zhiwei and Chen, Lin},
  title         = {A counterexample to the strong spin alignment conjecture},
  journal       = {arXiv preprint arXiv:2603.25410},
  year          = {2026},
  eprint        = {2603.25410},
  archivePrefix = {arXiv},
  primaryClass  = {quant-ph}
}

\begin{appendices}

\section{Comparison with Schatten-norm bounds}\label{app:normcomparison}
We compare the spectral implications of Theorem~\ref{thm:spectral} and the
Schatten-norm inequalities in Eq.~\eqref{eq:king-norm-comparison}.
The counterexamples below concern arbitrary states and are not assumed
to arise as channel outputs.

\subsection{Consequences and limitations of log-determinant ordering}
\label{app:entropy-comparison}
\begin{proposition}\label{prop:logdet-entropy}
Let $X$ and $Y$ be states on the same $D$-dimensional space. If
\begin{align}
 \mathcal L_t(X)\ge\mathcal L_t(Y),\qquad t>0,
 \label{eq:appendix-logdet-order}
\end{align}
then
\begin{align}
 \|X\|_q\le\|Y\|_q\quad(1<q\le2),
 \qquad S(X)\ge S(Y).
 \label{eq:appendix-entropy-consequence}
\end{align}
\end{proposition}
\begin{proof}
For $1<q<2$, define the positive constant $c_q$ by
\begin{align}
 c_q^{-1}:=\int_0^\infty t^{q-2}
       \bigl[1-\phi_t(1)\bigr]dt.
 \label{eq:moment-integral-constant}
\end{align}
The integral is finite because its integrand is $O(t^{q-2})$ at zero
and $O(t^{q-3})$ at infinity. A change of variable gives, for $a\ge0$,
\begin{align}
 a^q=c_q\int_0^\infty t^{q-2}
             \bigl[a-\phi_t(a)\bigr]dt.
 \label{eq:moment-scalar-integral}
\end{align}
Applying this identity to the eigenvalues of a state $Z$ yields
\begin{align}
 \Tr Z^q=c_q\int_0^\infty t^{q-2}
             \bigl[1-\mathcal P_t(Z)\bigr]dt.
 \label{eq:moment-logdet-integral}
\end{align}
Comparing the nonnegative integrands using
Eqs.~\eqref{eq:appendix-logdet-order} and~\eqref{eq:perspective-definition} gives
$\Tr X^q\le\Tr Y^q$ for $1<q<2$.
Continuity as $q\to2^-$ includes $q=2$.
Since both traces equal one at $q=1$,
\begin{align}
 \Tr(X\ln X)
 =\lim_{q\to1^+}\frac{\Tr X^q-1}{q-1}
 \le\lim_{q\to1^+}\frac{\Tr Y^q-1}{q-1}
 =\Tr(Y\ln Y).
 \label{eq:entropy-from-moment-limit}
\end{align}
The limits hold termwise on the finite spectra, with $0\ln0:=0$.
Multiplication by $-1/\ln2$ proves the entropy comparison.
\end{proof}

\paragraph{Higher norms.}
The conclusion for $1<q\le2$ does not extend to all larger exponents.
Consider the three-dimensional diagonal states
\begin{align}
 X:=\operatorname{diag}\left(\frac{42}{100},\frac{29}{100},\frac{29}{100}\right),
 \qquad
 Y:=\operatorname{diag}\left(\frac{40}{100},\frac{35}{100},\frac{25}{100}\right).
 \label{eq:logdet-higher-norm-example}
\end{align}
Direct expansion gives
\begin{align}
 \det(tI_3+X)-\det(tI_3+Y)
 =\frac{t}{5000}+\frac{161}{500000}>0,
 \qquad t>0.
 \label{eq:logdet-higher-norm-determinants}
\end{align}
Thus Eq.~\eqref{eq:appendix-logdet-order} holds strictly for every $t>0$.
Nevertheless,
\begin{align}
 \Tr X^3=\frac{61433}{500000}
 >\frac{49}{400}=\Tr Y^3,
 \label{eq:logdet-higher-norm-cubes}
\end{align}
so $\|X\|_3>\|Y\|_3$.

\subsection{Why Schatten norms alone cannot give a vanishing defect}
\label{app:norm-defect}
We retain ${\mathcal D}_{p,d}^{\otimes n}(\ket{0^n}\bra{0^n})$ as the reference state.
The following construction satisfies every Schatten-norm bound while
its spectral defect remains bounded away from zero.

\begin{proposition}[Schatten-norm bounds with nonvanishing spectral defect]\label{prop:norm-nonvanishing-defect}
Fix $d\ge2$ and $0<p\le p_{\max}(d)$ with $p\ne1$.
For all sufficiently large $n$, there exist
full-rank states $W_n$ on $(\mathbb C^d)^{\otimes n}$ such that
\begin{align}
 \|W_n\|_q&\le\bigl\|{\mathcal D}_{p,d}^{\otimes n}(\ket{0^n}\bra{0^n})\bigr\|_q,
 \qquad 1\le q\le\infty,
 \label{eq:all-norms-nonvanishing-defect}
\end{align}
and
\begin{align}
   \delta_{\operatorname{maj}}\!\left(W_n,{\mathcal D}_{p,d}^{\otimes n}(\ket{0^n}\bra{0^n})\right)
 &\to\frac12.
\end{align}
\end{proposition}
\begin{proof}
Let $D:=d^n$ and $\tau:={\mathcal D}_{p,d}(\ket{0}\bra{0})$, so that
$S(\tau)=-\alpha_p\log \alpha_p-(d-1)\beta_p\log \beta_p$.
Put $r_n:=\lceil2^{nS(\tau)}\rceil$ and choose a rank-$r_n$ projection $P_n$
on the output space. This is possible for all sufficiently
large $n$ because $0<S(\tau)<\log d$.
First consider $U_n:=P_n/r_n$.
For the surprisal sum $S_n$ in the proof of Theorem~\ref{thm:approximation},
specialized to $\sigma=\tau$,
$\E S_n=nS(\tau)$ and $\Var S_n=nv_{d,p}>0$.
Jensen's inequality gives, for every $q>1$,
\begin{align}
 \Tr(\tau^{\otimes n})^q
 =\E 2^{-(q-1)S_n}
 \ge2^{-(q-1)nS(\tau)}
 \ge r_n^{1-q}
 =\Tr U_n^q.
 \label{eq:flat-state-norms}
\end{align}
At $q=1$ both norms are one. The operator-norm inequality follows
by letting $q\to\infty$.

We next show that $\mathcal K_{r_n}(\tau^{\otimes n})\to1/2$.
There are at most $r_n$ eigenvalues at least $1/r_n$.
Among the largest $r_n$ eigenvalues, those below $1/(nr_n)$ have
total mass at most $1/n$. Thus
\begin{align}
 \Pr\{S_n\le\log r_n\}
 \le \mathcal K_{r_n}(\tau^{\otimes n})
 \le \Pr\{S_n\le\log r_n+\log n\}+\frac1n.
 \label{eq:entropy-rank-partial-sum}
\end{align}
Since $\log r_n=nS(\tau)+o(1)$ and $\log n=o(\sqrt n)$,
both probabilities converge to $1/2$ by the central limit theorem.
For $0\le u\le r_n$, the function
$\mathcal K_u(U_n)-\mathcal K_u(\tau^{\otimes n})=u/r_n-\mathcal K_u(\tau^{\otimes n})$
is convex and hence attains its maximum at an endpoint.
For $r_n\le u\le D$, it equals $1-\mathcal K_u(\tau^{\otimes n})$
and is nonincreasing. Consequently
\begin{align}
 \delta_{\operatorname{maj}}(U_n,\tau^{\otimes n})
 =1-\mathcal K_{r_n}(\tau^{\otimes n})\longrightarrow\frac12.
 \label{eq:flat-state-defect}
\end{align}

To obtain full rank, take any $0<\zeta_n<1$ with $\zeta_n\to0$ and set
\begin{align}
 W_n:=(1-\zeta_n)U_n+\zeta_n\pi_D.
 \label{eq:full-rank-flat-state}
\end{align}
Because $r_n\le D$, one has $\|\pi_D\|_q\le\|U_n\|_q$ for
$1\le q\le\infty$.
Convexity of the Schatten norms therefore gives
\begin{align}
 \|W_n\|_q
 \le(1-\zeta_n)\|U_n\|_q+\zeta_n\|\pi_D\|_q
 \le\|U_n\|_q\le\|\tau^{\otimes n}\|_q.
 \label{eq:full-rank-norms}
\end{align}
The uniform component preserves the eigenvalue ordering, so
\begin{align}
 \mathcal K_u(W_n)=(1-\zeta_n)\mathcal K_u(U_n)+\zeta_nu/D,
 \qquad 0\le u\le D.
 \label{eq:full-rank-partial-sums}
\end{align}
It follows that $|\mathcal K_u(W_n)-\mathcal K_u(U_n)|\le\zeta_n$ uniformly in $u$,
and hence
\begin{align}
 \bigl|\delta_{\operatorname{maj}}(W_n,\tau^{\otimes n})
       -\delta_{\operatorname{maj}}(U_n,\tau^{\otimes n})\bigr|\le\zeta_n.
 \label{eq:full-rank-defect-continuity}
\end{align}
This completes the proof.
\end{proof}

\section{Sharpness of the approximate-majorization bound}\label{app:sharpness}

\begin{proposition}\label{prop:approximation-sharpness}
For $\sigma=\operatorname{diag}(2/3,1/3)$, there exist $\rho^n$ satisfying
$\mathcal L_t(\rho^n)\ge\mathcal L_t(\sigma^{\otimes n})$ for every $t>0$
and all sufficiently large $n$, with
\begin{align}
 \delta_{\operatorname{maj}}(\rho^n,\sigma^{\otimes n})
 =\Omega\!\left(\frac{\log\log n}{\sqrt n}\right).
 \label{eq:sharpness-lower-bound}
\end{align}
\end{proposition}
\begin{proof}
The eigenvalues of $\sigma^{\otimes n}$ are $y_j=2^{n-j}/3^n$, with
multiplicities $\binom nj$, for $0\le j\le n$. Set
\begin{align}
 m=\lfloor n/3\rfloor,\qquad
 r=m+\left\lfloor\tfrac12\log\log n\right\rfloor.
 \label{eq:sharpness-cutoffs}
\end{align}
In the same eigenbasis, we form $\rho^n$ by replacing the eigenvalues with
$j\le m$ by $a$, those with $m<j\le r$ by $b:=y_r$, and leaving the
others unchanged. The value of $a$ is chosen to preserve the trace,
so $a>b>0$. Write $k:=\sum_{j=0}^m\binom nj$ and
$q:=\binom nm y_m$. Since $-\log y_j=n\log 3-n+j$ and
$\binom nj y_j$ is maximized at $j=m$, we have
$q=Q(S_{\sigma^{\otimes n}},1)$.
Stirling's formula gives $q=(3+o(1))/(2\sqrt{\pi n})$ and
$\binom nj y_j\ge q/2$ for $m\le j\le r$ and all sufficiently large $n$.
Consequently,
\begin{align}
 \delta_{\operatorname{maj}}(\rho^n,\sigma^{\otimes n})
 &\ge ka-\sum_{j=0}^m\binom nj y_j
 =\sum_{j=m+1}^r\binom nj(y_j-b)
 \ge\frac{r-m-1}{4}\,q,
 \label{eq:sharpness-one-rank}
\end{align}
where $y_j-b\ge y_j/2$ for $j<r$. This proves the required lower bound.

It remains to verify the log-determinant hypothesis. The $m$th term
of $k$ and Eq.~\eqref{eq:threshold-cardinality} give $q\le ky_m\le2q$.
Also, $ka\ge\sum_{j=0}^m\binom nj y_j\to1/2$ by the binomial central
limit theorem. For $t\ge0$, put
$G(t):=\mathcal L_t(\rho^n)-\mathcal L_t(\sigma^{\otimes n})$.
Integrating the same counting bound yields
\begin{align}
 \frac{G(0)}{k}
 &=\ln\frac ab-\frac1k\int_b^\infty
       \frac{N_{\sigma^{\otimes n}}(u)}{u}\,du
 \ge\ln\frac ab-\frac{2q}{kb}
 \ge\frac12\ln n-O\!\left(\sqrt{\ln n}\right)>0,
 \label{eq:sharpness-zero-shift}
\end{align}
since $b=2^{m-r}y_m$ and $2^{r-m}=O(\sqrt{\ln n})$.
Trace preservation further gives
\begin{align}
 \frac{d}{dt}\!\left[(t+a)^2(t+b)G'(t)\right]
 =-\sum_{j=0}^r\binom nj
   \frac{(y_j-b)(y_j-a)^2}{(t+y_j)^2}\le0.
 \label{eq:sharpness-shift-identity}
\end{align}
Hence $G'$ can change sign only from positive to negative, so $G$ is
bounded below by its endpoint values $G(0)>0$ and
$\lim_{t\to\infty}G(t)=0$. Thus $G(t)\ge0$ for every $t>0$.
\end{proof}

\section{A logarithmic second-moment bound}
\label{app:truncated-moment}

\begin{lemma}[Logarithmic second moment]
\label{lem:output-truncated-moment}
Fix $d\ge2$ and $0<p\le p_{\max}(d)$ with $p\ne1$.
There exists $B_{d,p}>0$ such that, for every $n\ge1$,
every state $\rho^n$ on $(\mathbb C^d)^{\otimes n}$,
and every $u>d^{-n}$,
\begin{align}
 u\Tr\!\left[
 \ln^2\!\left(I+\frac{\mathcal D_{p,d}^{\otimes n}(\rho^n)}{u}\right)
 \right]\ln(d^nu)\le B_{d,p}\sqrt n.
 \label{eq:output-truncated-moment}
\end{align}
\end{lemma}

\begin{proof}
Let $X={\mathcal D}_{p,d}^{\otimes n}(\rho^n)$.
Let $R_i$ replace the $i$th subsystem by $I_d/d$, leaving its complement
unchanged, and set $G_n:=\sum_{i=1}^n(\id-R_i)$.
For $u>d^{-n}$, put $F:=\ln(I+X/u)$.
We first bound the logarithmic second moment:
\begin{align}
 u\Tr F^2\ln(d^nu)
 &\le u\Tr F^2\ln\frac{d^n}{\Tr F^2}\\
 &=u\left[\Tr(F^2\ln F^2)
       -\Tr F^2\ln\frac{\Tr F^2}{d^n}\right]
       -u\Tr(F^2\ln F^2)\\
 &\le u\left[\Tr(F^2\ln F^2)
       -\Tr F^2\ln\frac{\Tr F^2}{d^n}\right]+\frac2e
 \label{eq:log-sobolev-moment-bound}\\
 &\le2(11+5\ln d)u\Tr[F G_n(F)]+\frac2e.
 \label{eq:logarithmic-moment-energy}
\end{align}
The first inequality uses $\Tr F^2\le1/u$, which follows from
$\ln^2(1+r)\le r$ for $r\ge0$.
The second follows from $f^2\ln f^2\ge-2f/e$ for $f\ge0$
and $u\Tr F\le1$, where the latter uses $\ln(1+r)\le r$.
The third is the tensor-product logarithmic Sobolev inequality
of \cite[Theorem~9]{TPK14}: the one-site generator
$A\mapsto\Tr(A)I_d/d-A$ has spectral gap one and stationary state $I_d/d$,
so its product generator has logarithmic Sobolev constant at least
$1/(11+5\ln d)$.
In the conventions of \cite[Eqs.~(6)--(8)]{TPK14}, the entropy is the
bracketed expression divided by $2d^n$, while the Dirichlet form is
$d^{-n}\Tr[F G_n(F)]$.
The inequality extends to positive semidefinite $F$ by continuity,
with $0\ln0:=0$.

We next bound the Dirichlet-form term $u\Tr[F G_n(F)]$.
In an eigenbasis $\{\ket j\}$ of $X$, with eigenvalues $x_j$,
define $m^{(i)}_{j\ell}:=\langle j|R_i(\ket\ell\bra\ell)|j\rangle$.
Then
\begin{align}
 u\Tr[F G_n(F)]
 &\le\frac12\sum_{i,j,\ell}m_{j\ell}^{(i)}
 (x_j-x_\ell)
 \left(\frac{x_j}{u+x_j}-\frac{x_\ell}{u+x_\ell}\right)\\
 &=\Tr[H G_n(X)]
 \le\frac12\|G_n(X)\|_1,
 \label{eq:clipping-energy-bound}
\end{align}
where $H=X(uI+X)^{-1}$.
The first inequality applies the scalar bound
\begin{align}
u[\ln(1+x/u)-\ln(1+y/u)]^2
\le (x-y)\left[\frac{x}{u+x}-\frac{y}{u+y}\right]
\end{align}
term by term, which follows from Cauchy--Schwarz.
The coefficients $m^{(i)}_{j\ell}$ are nonnegative and doubly
stochastic because $R_i$ is positive, unital, and trace preserving;
their symmetry follows from the Hilbert--Schmidt self-adjointness
of $R_i$. These properties give the equality by expanding $G_n$.
The last inequality uses $0\le H\le I$ and the fact that
$G_n(X)$ is Hermitian and traceless.

Combining Eqs.~\eqref{eq:logarithmic-moment-energy} and
\eqref{eq:clipping-energy-bound} gives
\begin{align}
 u\Tr F^2\ln(d^nu)
 \le(11+5\ln d)\|G_n(X)\|_1+\frac2e.
 \label{eq:clipping-moment-estimate}
\end{align}
It remains to bound $\|G_n(X)\|_1$ using the depolarizing structure.
For $0<p<1$, differentiation gives
$G_n(X)=-(1-p)\frac{d}{dp}\mathcal D_{p,d}^{\otimes n}(\rho^n)$.
Expanding the output over replacement sets yields
\begin{align}
 \|G_n(X)\|_1
 &\le\frac1p\sum_{k=0}^n\binom nk p^k(1-p)^{n-k}|k-np|\\
 &\le\frac1p\left[
       \sum_{k=0}^n\binom nk p^k(1-p)^{n-k}(k-np)^2
       \right]^{1/2}\\
 &=\sqrt{\frac{n(1-p)}p}.
 \label{eq:noise-sensitivity-binomial}
\end{align}
For the first inequality, each replacement set $J$ contributes a state
with weight $p^{|J|}(1-p)^{n-|J|}$; multiplying the derivative of this
weight by $-(1-p)$ gives the same weight times $(np-|J|)/p$.
The triangle inequality and grouping by $|J|$ give the first bound.
The second is the Cauchy--Schwarz inequality.

For $1<p\le p_{\max}(d)$, put
$\widetilde p:=d-(d-1)p\in[d/(d+1),1)$ and
$W_-(A):=(\Tr(A)I_d-A^T)/(d-1)$.
The map $W_-$ is a channel as its Choi matrix is non-negative.
Direct substitution gives
\begin{align}
 X=\mathcal D_{\widetilde p,d}^{\otimes n}
   \bigl(W_-^{\otimes n}((\rho^n)^T)\bigr).
 \label{eq:negative-depolarizing-factorization}
\end{align}
Here the transpose is global.
Applying Eq.~\eqref{eq:noise-sensitivity-binomial} at $\widetilde p$
therefore gives $\|G_n(X)\|_1\le b_{d,p}\sqrt n$ in both parameter ranges,
with $b_{d,p}>0$ depending only on $d,p$.
Equation~\eqref{eq:clipping-moment-estimate} and $n\ge1$ now prove the
claim with $B_{d,p}:=(11+5\ln d)b_{d,p}+2/e$.
\end{proof}

\section{Extension to unital qubit channels}\label{app:unital-qubit}

The argument of Theorem~\ref{thm:spectral} extends to unital qubit
channels using King's phase-damping decomposition~\cite{Kin02Unital}.
Once the log-determinant ordering is established, the general
approximate-majorization theorem and the coding arguments give a
third-order expansion with an $O(\log\log n)$ remainder.

Let $\mathcal N$ be a unital qubit channel with Bloch matrix $T$, so that
$\mathcal N((I_2+\boldsymbol r\cdot\boldsymbol\sigma)/2)
=(I_2+(T\boldsymbol r)\cdot\boldsymbol\sigma)/2$,
where $\boldsymbol\sigma$ is the vector of Pauli matrices.
Set
\begin{align}
 p:=1-\|T\|_{\mathrm{op}}\in[0,1],\qquad
 \tau:={\mathcal D}_{p,2}(\ket{0}\bra{0})
      =\alpha_p\ket{0}\bra{0}+\beta_p\ket{1}\bra{1},
 \label{eq:unital-parameters}
\end{align}
where $\|T\|_{\mathrm{op}}$ is the Euclidean operator norm and all
parameters from Section~\ref{sec:notation} are evaluated at $d=2$.
Thus $\alpha_p=1-p/2$ and $\beta_p=p/2$.

\begin{proposition}[Log-determinant ordering for unital qubit outputs]
\label{prop:unital-logdet}
For every $n\ge1$, every state $\rho^n$ on $(\mathbb C^2)^{\otimes n}$,
and every $t>0$,
\begin{align}
 \mathcal L_t\bigl(\mathcal N^{\otimes n}(\rho^n)\bigr)
 \ge\mathcal L_t\bigl(\tau^{\otimes n}\bigr).
 \label{eq:unital-logdet}
\end{align}
\end{proposition}

\begin{proof}
For $p=0$, the reference is pure; for $p=1$, the channel is completely
depolarizing. Both cases are immediate. Assume $0<p<1$.
The case $n=1$ follows because every output Bloch vector has length at
most $1-p$, so its spectrum is majorized by $(\alpha_p,\beta_p)$.

Assume the claim for $n-1$, and let $\rho_A$ be the first-input
marginal of $\rho^n$. King's decomposition~\cite{Kin02Unital}, with
input unitaries absorbed into the phase-damping bases, gives
probabilities $c_a$, unitaries $V_a$, and orthonormal bases
$\{\ket{u_i^{(a)}}\}_{i=1}^2$ such that
\begin{align}
 \mathcal N(X)
 &=\sum_a c_aV_a^\dagger\Delta_{p,2}^{(a)}(X)V_a,
 \qquad \text{for every }X,
 \label{eq:unital-king}\\
 \langle u_i^{(a)}|\rho_A|u_i^{(a)}\rangle
 &=\frac12,\qquad i=1,2,\quad\text{for every }a.
 \label{eq:unital-balance}
\end{align}
Here $\Delta_{p,2}^{(a)}$ is defined by Eq.~\eqref{eq:phase} in the
$a$th basis, and the decomposition may depend on $\rho_A$.
Lemma~\ref{lem:block} applies to any orthonormal basis by unitary conjugation.
The conditional inputs
\begin{align}
 \rho_{i,a}^{n-1}:=2(\bra{u_i^{(a)}}\otimes I)\rho^n
 (\ket{u_i^{(a)}}\otimes I)
 \label{eq:unital-conditional-states}
\end{align}
are normalized by Eq.~\eqref{eq:unital-balance} and give the
conditional-output relation in Eq.~\eqref{eq:conditional-output-blocks}
with $d=2$ and $\mathcal D_{p,2}$ replaced by $\mathcal N$.
The remaining concavity, block-determinant, and scaled-induction steps
in the proof of Theorem~\ref{thm:spectral} apply unchanged,
completing the induction.
\end{proof}

Let $M_{\mathcal N}^*(n,\varepsilon)$ denote the maximal message size
for $\mathcal N^{\otimes n}$ in the coding setting of
Section~\ref{sec:auxiliary-lemmas}.

\begin{corollary}[Approximate majorization and third-order classical capacity]
\label{cor:unital-capacity}
Suppose $0<p<1$.
The bound in Eq.~\eqref{eq:iid-approximation-bound} holds with
$\mathcal N^{\otimes n}(\rho^n)$ in place of $\rho^n$ and
$\sigma=\tau$. Its constant and starting blocklength depend only on $p$.
For every fixed $\varepsilon\in(0,1)$,
\begin{align}
 \log M_{\mathcal N}^*(n,\varepsilon)
 =nc_{2,p}+\sqrt{nv_{2,p}}\,\Phi^{-1}(\varepsilon)
       +\frac12\log n+O_{p,\varepsilon}(\log\log n),
 \label{eq:unital-third-order}
\end{align}
where $c_{2,p}$ and $v_{2,p}$ are given in Eq.~\eqref{eq:parameters}
with $d=2$. Arbitrary entangled codewords and collective decoding are
allowed.
\end{corollary}

\begin{proof}
Since $\tau={\mathcal D}_{p,2}(\ket{0}\bra{0})$ and
$V_\tau=v_{2,p}>0$, Proposition~\ref{prop:unital-logdet} and
Theorem~\ref{thm:approximation} give the approximate-majorization bound
with a constant depending only on $p$.
Lemma~\ref{lem:testing} then gives the i.i.d. testing converse in
Eq.~\eqref{eq:converse-testing}, now with an error-tolerance shift
of order $O_p(\log\log n/\sqrt n)$. The testing expansion in
Section~\ref{sec:mainproof} therefore yields a converse with an
$O_{p,\varepsilon}(\log\log n)$ remainder.

For achievability, encode into the two orthogonal pure states along a
maximal right singular direction of $T$. Their output Bloch vectors
are opposite and have length $1-p$. Measurement in their common
eigenbasis therefore induces exactly $W_{p,2}$ in
Eq.~\eqref{eq:classical}. The achievability argument in
Section~\ref{sec:mainproof} applies to this same binary symmetric
channel, completing the proof.
\end{proof}
\end{appendices}


\end{document}